\documentclass[reqno]{amsart}

\usepackage[mathscr]{eucal}
\usepackage{amssymb,verbatim,graphicx,tikz}

\usepackage{latexsym,color}
\usepackage{amsmath}
\usepackage[all]{xy}

\newtheorem{theorem}{Theorem}[section]
\newtheorem{lemma}[theorem]{Lemma}

\newtheorem{rem}[theorem]{Remark}

\newtheorem{question}[theorem]{Question}

\newtheorem{defin}[theorem]{Definition}
\newtheorem{claim}[theorem]{Claim}

\newenvironment{definition}{\begin{defin}\rm}{\end{defin}}

\newcommand{\Wsub}[1]{\mathbf W\kern -2pt_{#1}}
\newcommand{\Wsubit}[1]{W\kern -2pt_{#1}}

\newcommand{\dcomp}{\mathbin{*}}
\newcommand{\D}{\mathsf{D}}
\newcommand{\R}{\mathsf{R}}
\newcommand{\pointsto}[1]{\stackrel{#1}{\longrightarrow}}
\newcommand{\out}{\operatorname{out}}

\newcommand\acomp{\mathbin{;}}
\newcommand\join{\mathbin{+}}
\newcommand\dom{\operatorname{\mathsf{D}}}
\newcommand\ran{\operatorname{\mathsf{R}}}
\newcommand\freevar{\mathcal{F}_{\textit{Var}}}
\newcommand\termvar{T_{\textit{Var}}}
\newcommand\var{\textit{Var}}

\newcommand{\axa}{\mathit{Ax^a}}
\newcommand{\axd}{\mathit{Ax^d}}

\newcommand{\isom}{\mathbb{I}}
\newcommand{\suba}{\mathbb{S}}

\newcommand{\vari}{\mathbb{V}}
\newcommand{\repr}{\mathbb{R}}

\newcommand{\up}[1]{\textup{#1}}

\tikzset{vertex/.style={draw, shape=circle, fill=black, inner sep=0pt, minimum size=3pt},
vertexe/.style={draw, shape=circle, fill=white, inner sep=0pt, minimum size=3pt},}

\begin{document}
\title{Domain and range for angelic and demonic compositions}

\author{Marcel Jackson}\thanks{The first author was supported by ARC Future Fellowship FT120100666 and Discovery Project DP1094578}
\address{Department of Mathematics and Statistics, La Trobe University, Melbourne VIC 3086, Australia}
\email{m.g.jackson@latrobe.edu.au}
\author{Szabolcs Mikul\'as}
\address{Department of Computer Science and Information Systems, Birkbeck College, University of London, London WC1E 7HX, UK}
\email{szabolcs@dcs.bbk.ac.uk}


\begin{abstract}
We give finite axiomatizations for the varieties generated by representable 
domain--range algebras when the semigroup operation is interpreted as
angelic or demonic composition, respectively.
\end{abstract}


\maketitle

\section{Introduction}
Any formal approach to modelling programs must encompass both logic and action.  On the one hand, the role of programs is to create and change input: an action on the state space.  On the other hand, the technical action of programs requires conditional tests that are logical in nature.  A common algebraic formalism for this is to model programs as relations on the state space and use restrictions of the identity relation to model logical propositions.  This is elegantly argued in the articles~\cite{DMS} and \cite{MS}, where it is observed that Kozen's axiom system KAT (Kleene algebra with tests) and the program logic PDL (Propositional Dynamic Logic) can be  unified by enriching the language of KAT with the introduction of unary operations modelling the domain and range of relations.   This enables the modal operations of dynamic logic to be precisely captured in a one-sorted algebraic setting.   

These and other articles provide simple axiomatic systems that are sound for the relational program semantics and which are sufficient to reason about many important aspects of programs.  
Completeness of these axioms seems more elusive.   Work involving the present authors showed that no finite system of axioms is sufficient to capture the full first order theory of the algebra of relations under composition with domain and/or range (amongst other operations) \cite{hirmik,jacsto:amon}.  This fact is just one of a swathe of negative results relating to the theory of binary relations.  For signatures involving composition, intersection and either of union or converse, not only is there no finite system of laws, but no complete system of laws can be recursively decided on finite algebras, see \cite[Theorem~2.5]{neu} and \cite[Theorem~8.1]{hirjac}.  Despite these negative results, in many situations it is sufficient to find systems that are complete for equations (rather than the full first order theory).  Kozen's system KA~\cite{koz} is precisely one such system that is complete for the equational theory of relations under composition, union and reflexive transitive closure; others include Andr\'eka \cite{and:91}, and Andr\'eka and Bredkhin \cite{AB-equ-95}.
One of the main results of the present article is to provide a relatively simple (and finite) system of equational axioms that is sound and complete for the equational theory of the operations of domain, range, composition and relational union: Theorem \ref{thm-main} below.

All of the above references consider program composition modelled as conventional composition of binary relations.  The validity of this approach depends on how one is to model nondeterminism: if programs are to be modelled as binary relations on the state space, then how should non-termination of programs be modelled? In particular, what if some computations terminate and some do not?  The most common approach is the so-called angelic model.  In the angelic model of a program $p$ as a relation on the state space, we consider $p$ to relate state $x$ to state $y$ if amongst the possible computations of $p$ when started at $x$ is one that terminates at $y$.  This does not preclude the possibility that some computations of $p$ at $x$ do not terminate.  A stricter model---the demonic model \cite{ngu}---requires in addition (to some computations of $p$ at $x$ leading to $y$), that \emph{all computations of $p$ at $x$ eventually terminate}\footnote{To the best of the authors' knowledge, the earliest reference to demons in this context may be in Broy, Gnatz and Wirsing \cite{BGW}, where Dijkstra \cite{dij} is cited as an instance of demonic modelling of nondeterminancy.}.

When programs are to be modelled angelically, then the relation associated to the composite of the programs $p$ and $q$ is just the usual relational composite of the relations corresponding to $p$ and $q$.  When programs are modelled demonically, then the relation associated to the composite of the programs is the demonic composition of the relations associated to $p$ and $q$.  Demonic composition as a binary operation on binary relations remains associative (see \cite{bacvdW} or \cite[\S5]{bersch} for example), but its general algebraic properties have seen far less algebraic consideration than its angelic counterpart.  A variant of the program logic PDL is developed and shown to be complete and decidable in \cite{demorl}, while in \cite{DJS}  it is observed that the algebra of binary relations under demonic composition and domain is indistinguishable from the algebra of partial maps under conventional (angelic) composition and domain.  This latter system has been well studied and has a well-known complete equational axiomatisation \cite{tro,jacsto:01,man}.   The algebraic properties of demonic relations and partial maps diverge once range information is incorporated.
The second main result of the present article is to find a simple axiomatic system that is sound and complete for the equational theory of demonic composition of relations with domain and range; Theorem~\ref{thm-maind}.

Almost all of the discussion and background to this point concerns formal systems for reasoning about programs and program correctness.  The results of the present article can be seen as part of a continued effort to embed these formal approaches in simple algebraic settings, where basic equational reasoning can be applied. 
Program modelling is just one of several motivations for this work however.  Domain and range are already transparently natural features of both relations and functions, and the modelling of these via unary operations can be traced back at least to the work of Menger \cite{men}, through the work of Schweizer and Sklar \cite{SS1,SS2,SS3,SS4}, Trokhimenko \cite{tro}, Bredikhin \cite{bre} and Schein \cite{sch} and into the work of the authors and their collaborators \cite{HJM,jacsto:01,jacsto:BS,jacsto:modal} as well as in the category-theoretic work of Cockett, Lack, Guo, Hofstra, Manes amongst others \cite{CGH1,CGH2,coclac,cocman}.
Yet another motivation comes from the structural theory of semigroups, where many authors have enriched the usual associative binary multiplication by the addition of unary operations that map onto idempotent elements; see for example Fountain \cite{fou:ade,fou:abu}, the work of Batbedat \cite{bat}, Lawson \cite{law} and many others; a survey on aspects of this theme of research can be found in Hollings \cite{hol}.  The axiom systems we investigate here appear as natural cases in this purely theoretical context. 

\section{Basics}
\begin{definition}\label{def:repdr}
Let $U$ be a set.
We define operations on elements of $\wp(U\times U)$.
\begin{description}
\item[Domain]
$$
\dom(X)=\{(u,u)\mid (u,v)\in X\text{ for some }v\in U\}
$$
\item[Range]
$$
\ran(X)=\{(v,v)\mid (u,v)\in X\text{ for some }u\in U\}
$$
\item[Angelic composition]
$$
X\acomp Y=\{(u,v)\mid (u,w)\in X\text{ and }(w,v)\in Y\text{ for some }w\in U\}
$$
\item[Demonic composition]
\begin{align*}
X\dcomp Y&=\{(u,v)\mid \text{for some }w\in U, (u,w)\in X\text{ and }(w,v)\in Y,\\
&\qquad\text{and for all }w\in U \text{ such that }(u,w)\in X, (w,w)\in\dom(Y)\}
\end{align*}
\end{description}
for every $X,Y\subseteq U\times U$.

The class $\repr(\acomp,\dom,\ran)$ of \emph{angelic domain--range semigroups} is
$$
\isom\suba \{(\wp(U\times U),\acomp,\dom,\ran)\mid U\text{ a set}\}
$$
while the class $\repr(\dcomp,\dom,\ran)$ of \emph{demonic domain--range semigroups} is
$$
\isom\suba \{(\wp(U\times U),\dcomp,\dom,\ran)\mid U\text{ a set}\}
$$
where $\isom$ and $\suba$ denote isomorphic copies and subalgebras, respectively.
\end{definition}

We may call elements of $\repr(\acomp,\dom,\ran)$ and $\repr(\dcomp,\dom,\ran)$
\emph{representable algebras}.
In this paper we give finite equational axiomatizations to the equational theories of
representable domain--range algebras.

A comment on notation: by default we use the letters $a,b,c,d,e,f,g,r,s,t$ for elements of algebras (which are typically to be represented as binary relations) and $p,q,u,v,w,x,y,z$ for abstract variables and for points on which relations act.  Inevitably, there is some entanglement of this convention in cases where algebras act on points formed from themselves and where algebras are formed from abstract variables.

\section{Angelic composition}

In this section we look at angelic composition.
We expand the signature of angelic domain--range semigroups with a join operation $\join$
that is interpreted as union, and define $\repr(\acomp,\dom,\ran,\join)$ as
$$
\isom\suba \{(\wp(U\times U),\acomp,\dom,\ran,\join)\mid U\text{ a set}\}.
$$
Using join $\join$ we can define an ordering $\le$ in the usual way:
$x\le y$ iff $x\join y = y$.

Let $\axa$ be the following set of equations:
\begin{align}
x\acomp(y\acomp z)&=(x\acomp y)\acomp z\label{eq-ass}\\
\dom(x)\acomp x&= x\label{eq-d}\\
x\acomp\ran(x)&=x\label{eq-r}\\
\dom(x)\acomp\dom(x)&=\dom(x)\label{eq-dcd}\\
\ran(x)\acomp\ran(x)&=\ran(x)\label{eq-rcr}\\
\dom(x\acomp y)&=\dom(x\acomp\dom(y))\label{eq-dcdcd}\\
\ran(x\acomp y)&=\ran(\ran(x)\acomp y)\label{eq-rcrrc}\\
\dom(\dom(x)\acomp y)&=\dom(x)\acomp\dom(y)\label{eq-ddcdcd}\\
\ran(x\acomp \ran(y))&=\ran(x)\acomp\ran(y)\label{eq-rcrrcr}\\
\dom(\ran(x))&=\ran(x)\label{eq-dr}\\
\ran(\dom(x))&=\dom(x)\label{eq-rd}\\
\dom(x)\acomp\dom(y)&=\dom(y)\acomp\dom(x)\label{eq-dcom}\\
\ran(x)\acomp\ran(y)&=\ran(y)\acomp\ran(x)\label{eq-rcom}\\
\dom(x)\acomp y&\le y\label{eq-dcord}\\
x\acomp \ran(y)&\le x\label{eq-crord}\\
\intertext{together with the axioms stating that join~$\join$ is a semilattice 
(idempotent, commutative and associative) operation
and that the operations are additive:}
x\acomp (y\join z) &= x\acomp y \join x\acomp z\label{eq-ladd}\\
(x\join y)\acomp z &= x\acomp z \join y\acomp z\\
\dom(x\join y) &= \dom(x)\join \dom(y)\\
\ran(x\join y) &= \ran(x)\join \ran(y)
\end{align}
There was no attempt made to make the above axiom system independent.
Instead we aimed for symmetry and stated both the ``domain'' and ``range'' versions of the axioms.  The laws \eqref{eq-ass}--\eqref{eq-rcom} can very easily be shown equivalent to the \emph{closure semigroups with the left and right congruence conditions} in the sense of the first author and Stokes \cite{jacsto:01}.  The  \emph{adequate semigroups} of Fountain \cite{fou:ade} form a particularly well-studied special case; see Kambites~\cite{kam} for example.

Our main result about angelic composition is the following finite axiomatizability theorem.

\begin{theorem}\label{thm-main}
The variety $\vari(\acomp,\dom,\ran,\join)$ generated by 
the representation class
$\repr(\acomp,\dom,\ran,\join)$
is axiomatized by $\axa$.
\end{theorem}

\begin{proof}
Our task is to  show that, for all $(\acomp,\dom,\ran,\join)$-terms $s,t$,
$$
\vari(\acomp,\dom,\ran,\join)\models s=t\text{ iff }\axa\vdash s=t
$$
where $\vdash$ denotes derivability in equational logic.

The right-to-left direction follows by the validity of the axioms,
which can be easily checked.
For the other direction we have to show that
$\vari(\acomp,\dom,\ran,\join)\models s=t$ implies $\axa\vdash s=t$.
In fact we will show its contrapositive:
we will assume that $\axa\not\vdash s=t$
and construct a representable algebra $\mathcal{A}\in \repr(\acomp,\dom,\ran,\join)$
such that $\mathcal{A}\not\models s=t$.
The rest of this section is devoted to this task.
\end{proof}

After establishing some elementary consequences of the axioms in Section~\ref{sec-ep},
we introduce term graphs in Section~\ref{sec-tg}.
Using term graphs allows us to focus on join-free terms (Section~\ref{sec-ej}).
We will also use term graphs to show some special properties of the free algebra of $\axa$ 
in Section~\ref{sec-fa}.
In particular, we show that the free algebra is free from cycles:
for join-free terms $r,s,t$,
$\dom(r)\le s\acomp t$ implies $\dom(r)\le s=\dom(s)$ and $\dom(r)\le t=\dom(t)$. 
Observe that this is not true in general, since it is easy to construct (representable) algebras
with cycles:
e.g., let $s=\{(0,1)\}$ and $t=\{(1,0)\}$, whence 
$\dom(s)=\{(0,0)\}\le \{(0,1)\}\acomp\{(1,0)\}= s\acomp t$. 
However, the lack of cycles in the free algebra makes our task of
representing the free algebra easier (Section~\ref{sec-con}).

\subsection{Elementary properties}\label{sec-ep}
We start with the following easy consequences of the axioms.
First, $\le$ is indeed an ordering:
\begin{align}
x&\le x\label{eq-ordr}\\
x\le y \ \&\ y\le x&\Rightarrow x=y\label{eq-ords}\\
x\le y \ \&\ y\le z&\Rightarrow x\le z\label{eq-ordt}
\end{align}
Using the additivity of the operations we get that the operations are monotonic w.r.t.\ $\le$:
\begin{align}
x\le x'&\Rightarrow \dom(x)\le\dom(x')\label{eq-dmon}\\
x\le x'&\Rightarrow \ran(x)\le\ran(x')\label{eq-rmon}\\
x\le x'\ \& \ y\le y'&\Rightarrow x\acomp y\le x'\acomp y'\label{eq-cmon}
\end{align}

Let $\mathcal{A}=(A,\acomp,\dom,\ran,\join)$ be a model of $\axa$.
We extend the operations to sets of elements in the obvious way:
\begin{align*}
\dom(X)&=\{\dom(x)\mid x\in X\}\\
\ran(X)&=\{\ran(x)\mid x\in X\}\\
X\acomp Y&=\{x\acomp y\mid x\in X, y\in Y\}
\end{align*}
for every $X,Y\subseteq A$.
In particular, we define the set $\dom(A)$ of \emph{domain elements} of $\mathcal{A}$ as
\[
\dom(A)=\{\dom(a)\mid a\in A\}=\{a\in A\mid \dom(a)=a\}.
\]
Observe that range elements (defined analogously) coincide with domain elements,
since $\dom(\ran(a))=\ran(a)$ and $\ran(\dom(a))=\dom(a)$.


The following routine facts can likely be extracted from many previous articles (such as \cite{DJS}), but we give proofs as they are short, and because the veracity of our completeness result (Theorem \ref{thm-main}) depends on the facts following from the axioms precisely as presented here.
\begin{claim}\label{claim-sl}
Let $\mathcal{A}$ be a model of $\axa$.
\begin{enumerate}
\item\label{it-one}
The algebra $(\dom(A),\acomp)$ of domain elements is a (lower) semilattice and
the semilattice ordering 
coincides with $\le$.
\item\label{it-two}
For every $a\in A$,
$\dom(a)$ (resp.\ $\ran(a)$) is the minimal element $d$ in $\dom(A)$ such that
$d\acomp a=a$ (resp.\ $a\acomp d=a$).
\item\label{it-three}
For every $a\in A$ and $d,e\in\dom(A)$,
we have $d\acomp a\acomp e\le a$.
\end{enumerate}
\end{claim}

\begin{proof}
\ref{it-one}:
By~\eqref{eq-ddcdcd} the set of domain elements is closed under~$\acomp$,
which is an associative~\eqref{eq-ass}, idempotent~\eqref{eq-d} and
commutative~\eqref{eq-dcom} operation on domain elements.
The semilattice ordering is defined by
$$
\dom(x)\le '\dom( y)\text{ iff } \dom(x)\acomp \dom(y)=\dom(x)
$$
and we claim that this is equivalent to 
the definition of $\dom(x)\le\dom(y)$ by $\dom(x)\join\dom(y)=\dom(x)$.
Assuming $\dom(x)\acomp \dom(y)=\dom(x)$ we have
$\dom(x)\join\dom(y)=\dom(x)\acomp\dom(y)\join\dom(y)=\dom(y)$,
since $\dom(x)\acomp\dom(y)\le\dom(y)$ (by~\eqref{eq-dcord}).
Assuming $\dom(x)\join\dom(y)=\dom(y)$ we get
\begin{align*}
\dom(x)\acomp\dom(y)&=\dom(x)\acomp(\dom(x)\join\dom(y))&&\\
&=\dom(x)\acomp\dom(x)\join\dom(x)\acomp\dom(y)&&\text{by~\eqref{eq-ladd}}\\
&=\dom(x)\join\dom(x)\acomp\dom(y)&&\text{by~\eqref{eq-dcd}}\\
&=\dom(x)&&
\end{align*}
since $\dom(x)\ge \dom(x)\acomp\ran(\dom(y))=\dom(x)\acomp\dom(y)$
by~\eqref{eq-crord} and~\eqref{eq-rd}.

\ref{it-two}:
Assume that $\dom(d)\acomp a= a$. Then
\begin{align*}
\dom(a)&=\dom(\dom(d)\acomp a)&&\\
&=\dom(d)\acomp \dom(a)&&\text{by~\eqref{eq-ddcdcd}}\\
&\le\dom(d)&&
\end{align*}
as desired.
The proof of the statement about range is analogous.

\ref{it-three}:
Straightforward by~\eqref{eq-dcord} and~\eqref{eq-crord}.
\end{proof}


\subsection{Term graphs}\label{sec-tg}
Let $\termvar^-$ be the set of $(\acomp,\dom,\ran)$-terms generated by the set of variables $\var$.
We will refer to these as join-free terms.
We will adopt the concept of {term graphs} from Andr\'eka and Bredikhin~\cite{AB-equ-95} to~$\termvar^-$.
The operations~$\dom$ and~$\ran$ are not considered explicitly in~\cite{AB-equ-95},
but the concepts and proofs are easily modified to cover these as well.

A \emph{labelled graph} is a structure $G=(V,E)$
where $V$ is a set of vertices and $E\subseteq V\times \var\times V$ is a set of labelled edges.
Given two labelled graphs $G_1=(V_1,E_1)$ and $G_2=(V_2,E_2)$,
a {\em homomorphism} $h\colon G_1\to G_2$ is a map from $V_1$ to $V_2$ that preserves
labelled edges: if $(u,x,v)\in E_1$, then $(h(u), x,h(v))\in E_2$.
Given an equivalence relation $\theta$ on $V$, the
{\em quotient graph} is $G/\theta=(V/\theta,E/\theta)$ where
$V/\theta$ is the set of equivalence classes of $V$ and
$$
E/\theta =\{ (u/\theta, x,v/\theta):(u,x,v)\in E\text{ for some }
u\in u/\theta\text{ and }v\in v/\theta\}.
$$
A \emph{2-pointed graph} is a labelled graph $G=(V,E)$ with two
(not necessarily distinct) distinguished vertices $\iota, o\in V$.
We will call $\iota$ the \emph{input} and $o$ the \emph{output} vertex
of $G$, respectively, and denote 2-pointed graphs as $G=(V,E,\iota, o)$.
In the case of 2-pointed graphs, we require that a homomorphism
preserves input and output vertices as well. 

Let $G_1\oplus G_2$ denote the disjoint union of $G_1$ and $G_2$.
For 2-pointed graphs $G_1=(V_1,E_1,\iota_1,o_1)$ and $G_2=(V_2,E_2,\iota_2,o_2)$,
we define their \emph{composition} as
$$
G_1\acomp G_2=(((V_1,E_1)\oplus(V_2,E_2))/\theta,\iota_1/\theta,o_2/\theta)
$$
where $\theta$ is the smallest equivalence relation on the
disjoint union $V_1$ and $V_2$
that identifies $o_1$ with $\iota_2$.
When no confusion is likely we will identify
an equivalence class $u/\theta$ with $u$, hence
$\iota_i/\theta$ with $\iota_i$ and
$o_i/\theta$ with $o_i$ for $i\in\{1,2\}$.

We define {\em term graphs} as special 2-pointed graphs
by induction on the complexity of terms.
For variable $x$, we choose distinct points $\iota, o$ and let
$$
G_{x}=(\{\iota, o\},\{(\iota,x, o)\},\iota,o).
$$
Let $s$ be a term and assume that $G_s=(V,E,\iota,o)$.
We define
$$
G_{\dom(s)}=(V,E,\iota,\iota)\text{ and }
G_{\ran(s)}=(V,E,o,o).
$$
Finally, for terms $s$ and $t$, we set
$$
G_{s\acomp t}=G_s\acomp G_t.
$$


For any term $s$, we can consider $G_s=(V_s,E_s,\iota_s, o_s)$ as a representable algebra.
To this end let $\sharp\colon \var \to \wp(E_s)$ be a valuation of variables such that 
$$
x^\sharp=\{(u,v)\in V_s\times V_s\mid (u,x,v)\in E_s\}
$$
for every variable $x$ occurring in $s$ (notation $x\in s$).
We define the representable algebra $\mathcal{G}_s$ as
the subalgebra of $(\wp(V_s\times V_s),\acomp,\dom,\ran,\join)$ generated by 
$\{x^\sharp\mid x\in s \}$.
The \emph{universe} $W_s$ of $\mathcal{G}_s$ is
the reflexive--transitive closure of $\bigcup \{x^\sharp\mid x\in s\}$.
Observe that, by the construction of the term graph $G_s$,
$W_s$ is an antisymmetric relation.

We extend $\sharp$ to an interpretation of complex terms in the obvious way:
$$
(\dom(t))^\sharp=\dom(t^\sharp)\quad
(\ran(t))^\sharp=\ran(t^\sharp)\quad
(t_1\acomp t_2)^\sharp=t_1^\sharp\acomp t_2^\sharp
$$
for terms $t,t_1,t_2$.
By an easy induction on the complexity of terms we get the following.
\begin{claim}\label{claim-true}
In $\mathcal{G}_s$, we have $(\iota_s,o_s)\in s^\sharp$.
\end{claim}

Next we recall a characterization of validities using
graph homomorphisms from \cite[Theorem~1]{AB-equ-95}.
\begin{theorem}\label{thm-ab}
The inequality $s\le t$ is valid in representable algebras
iff there is a homomorphism from $G_t$ to $G_s$.
\end{theorem}
Observe that $s\le t$ implies that all the variables in $t$ must occur in $s$.
The key step in proving Theorem~\ref{thm-ab} is the following lemma,
see \cite[Lemma~3]{AB-equ-95}.

\begin{lemma}\label{lem-ab3}
Let $s,t$ be terms and consider $G_s=(V_s,E_s,\iota_s, o_s)$.
Then $(\iota_s,o_s)\in t^\sharp$ iff
there is a homomorphism from $G_t$ to $G_s$.
\end{lemma}

\subsection{Eliminating join}\label{sec-ej}
Recall that our task is to show that for any $(\acomp,\dom,\ran,\join)$-terms $s,t$,
$$
\vari(\acomp,\dom,\ran,\join)\models s=t\text{ implies }\axa\vdash s=t.
$$
This is obviously equivalent to the statements
\begin{align*}
\vari(\acomp,\dom,\ran,\join)\models s\le t\text{ implies }\axa\vdash s\le t,\\
\vari(\acomp,\dom,\ran,\join)\models s\ge t\text{ implies }\axa\vdash s\ge t.
\end{align*}
Thus we assume that $\vari(\acomp,\dom,\ran,\join)\models s\le t$ and
want to show $\axa\vdash s\le t$.

Next we show that the above can be reduced to join-free terms.
Since the operations are additive, every term $s$ can be equivalently written 
in the form $s_1\join\ldots \join s_n$ for some join-free terms $s_1,\ldots , s_n$,
whence we have
$$
\vari(\acomp,\dom,\ran,\join)\models s_1\join\ldots\join s_n=s\le t= t_1\join\ldots \join t_m
$$
for some join-free terms $s_1,\ldots ,s_n,t_1,\ldots , t_m$.
Thus we have
$$
\vari(\acomp,\dom,\ran,\join)\models s_i\le t_1\join\ldots \join t_m
$$
for every $1\le i\le n$.
Recall the term graph $G_{s_i}=(V_{s_i},E_{s_i},\iota_{s_i}, o_{s_i})$
and let $\sharp\colon \var \to \wp(E_{s_i})$ be a valuation of the variables occurring in 
$s_i,t_1,\ldots ,t_m$ such that $x^\sharp=\{(u,v)\in V_{s_i}\times V_{s_i}\mid (u,x,v)\in E_{s_i}\}$
for the variables $x$ occurring in $s_i$.
Consider the generated algebra $\mathcal{G}_{s_i}$.
Since $\mathcal{G}_{s_i}$ is representable, we get
$$
\mathcal{G}_{s_i}\models s_i\le t_1\join\ldots \join t_m
$$
whence by Claim~\ref{claim-true}, we get
$$
(\iota_{s_i}, o_{s_i})\in s_i^\sharp
$$
whence
$$
(\iota_{s_i}, o_{s_i})\in t_1^\sharp\join\ldots\join t_m^\sharp
$$
and thus
$$
(\iota_{s_i}, o_{s_i})\in t_j^\sharp
$$
for some $1\le j\le m$.
By Lemma~\ref{lem-ab3}, there is a homomorphism from
$G_{t_j}$ to $G_{s_i}$.
By Theorem~\ref{thm-ab} we get that
$\vari(\acomp,\dom,\ran,\join)\models s_i\le t_j$.
Now if we manage to show that this implies $\axa\vdash s_i\le t_j$,
then we get
$$
\axa\vdash s=s_1\join\ldots\join s_n\le t_1\join\ldots\join t_m=t
$$
as desired.

Thus it suffices to show
$$
\vari(\acomp,\dom,\ran,\join)\models s\le t\text{ implies }\axa\vdash s\le t
$$
for join-free terms $s,t$.

\subsection{The free algebra}\label{sec-fa}

We will need some properties of the free algebra of the variety defined by $\axa$.

Let $\var$ be a countable set of variables and let $\freevar=(F_\var,\acomp,\dom,\ran,\join)$ 
be the free algebra of the variety defined by $\axa$ freely generated by $\var$.
Recall that the elements of $\freevar$ are the equivalence classes of terms,
where two terms $s$ and $t$ are {equivalent} 
if the equation $s=t$ is derivable from $\axa$ using equational logic.
When we want to emphasize the difference between a term $t$ and its equivalence class,
we will write $\overline{t}$ for the latter.
The operations in $\freevar$ are defined in the obvious way:
$$
\overline{t}\acomp\overline{s}=\overline{t\acomp s}\quad
\dom(\overline{t})=\overline{\dom( t)}\quad
\ran(\overline{t})=\overline{\ran( t)}\quad
\overline{t}\join\overline{s}=\overline{t\join s}
$$
and recall that this definition is indeed independent of the choice of terms $s,t$
from the equivalence classes $\overline{s},\overline{t}$.

\begin{claim}\label{claim-domain}
Let $s,t$ be join-free terms such that 
$\freevar\models \dom({s})\le t$.
Then $t$ is a composition of domain and range terms, i.e.,
it has the syntactical form $\dom(t_1)\acomp\ran(t_2)\acomp\ldots\acomp\dom(t_{n-1})\acomp\ran(t_n)$
for some terms $t_1,t_2,\ldots,t_{n-1},t_n$
\up(allowing some, but not all, of the terms being empty\up).
\end{claim}

\begin{proof}
Let $s,t$ be as in the claim.
Since the axioms are valid in representable algebras and the term graphs are representable, 
we get $\mathcal{G}_{\dom(s)}\models \dom(s)\le t$.
Towards a contradiction, assume that $t$ does not have the form  
$\dom(t_1)\acomp\ran(t_2)\acomp\ldots\acomp\dom(t_{n-1})\acomp\ran(t_n)$.
Then $t$ can be written in the form $r_1\acomp r_2\acomp\ldots\acomp r_m$
such that at least one $r_i$ is a variable, say, $x$.
Since the universe $W_{\dom(s)}$ of $\mathcal{G}_{\dom(s)}$ is an antisymmetric relation 
and there are no loops labelled by variables
(edges of the form $(u,x,u)$) in $G_{\dom(s)}$,
it follows that $(\iota_{\dom(s)},\iota_{\dom(s)})\notin t^\sharp$.
On the other hand,
by Claim~\ref{claim-true} we have 
$(\iota_{\dom(s)},\iota_{\dom(s)})\in (\dom(s))^\sharp$.
Thus $\mathcal{G}_{\dom(s)}\not\models \dom(s)\le t$. 
\end{proof}

\begin{claim}\label{claim-refl}
Let $r,s,t$ be join-free terms such that $\freevar\models \dom(r)\le s\acomp t$.
Then $\freevar\models \dom(r)\le s=\dom(s)$ and $\freevar\models \dom(r)\le t=\dom(t)$. 
\end{claim}

\begin{proof}
By Claim~\ref{claim-domain} we have that $s\acomp t$ is a composition of domain and range terms,
whence so are both of them. 
Since domain elements are closed under the operation~$\acomp$ (Claim~\ref{claim-sl}),
$s$ and $t$ are clearly domain elements when interpreted in the free algebra.
Thus $\freevar\models \dom(r)\le \dom(s)\acomp\dom(t)$.
Since $(\dom(F_\var),\acomp)$ is a semilattice with the ordering~$\le$ (Claim~\ref{claim-sl}),
we get $\freevar\models \dom(r)\le \dom(s)=s$ and $\freevar\models \dom(r)\le \dom(t)=t$. 
\end{proof}

\begin{claim}\label{claim-domain2}
Let $s,t$ be join-free terms such that 
$\freevar\models s\le \dom(t)$.
Then $s$ is a composition of domain and range terms, i.e.,
it has the syntactical form $\dom(s_1)\acomp\ran(s_2)\acomp\ldots\acomp\dom(s_{n-1})\acomp\ran(s_n)$
for some terms $s_1,s_2,\ldots,s_{n-1},s_n$
\up(allowing some, but not all, of the terms being empty\up).
\end{claim}

\begin{proof}
Let $s,t$ be as in the claim.
Since the axioms are valid in representable algebras and the term graphs are representable, 
we get $\mathcal{G}_{s}\models s\le \dom(t)$.
Towards a contradiction, assume that $s$ does not have the form  
$\dom(s_1)\acomp\ran(s_2)\acomp\ldots\acomp\dom(s_{n-1})\acomp\ran(s_n)$.
Then $s$ can be written in the form $r_1\acomp r_2\acomp\ldots\acomp r_m$
such that at least one $r_i$ is a variable, say, $x$.
Since the universe $W_{s}$ of $\mathcal{G}_{s}$ is an antisymmetric relation 
and there are no loops labelled by variables
(edges of the form $(u,x,u)$) in $G_{s}=G(r_1)\acomp\ldots\acomp G(r_m)$, 
it follows that $\iota_s\ne o_s$.
By Claim~\ref{claim-true} we have $(\iota_s,o_s)\in s^\sharp$.
On the other hand, since $\mathcal{G}_s$ is representable,
we have $(\iota_{s},o_{s})\notin (\dom(t))^\sharp$.
Thus $\mathcal{G}_{s}\not\models s\le \dom(t)$. 
\end{proof}

\begin{claim}\label{claim-refl2}
Let $s,t$ be join-free terms such that $\freevar\models s\le \dom(t)$.
Then $\freevar\models s=\dom(s)$.
\end{claim}

\begin{proof}
By Claim~\ref{claim-domain2} we have that $s$ is a composition of domain and range terms. 
Since domain elements are closed under the operation~$\acomp$ (Claim~\ref{claim-sl}),
$s$ is clearly a domain element when interpreted in the free algebra.
\end{proof}


\subsection{The construction}\label{sec-con}

Using the results of the previous sections,
we assume that, for some join-free terms $s,t$, we have  
$\vari(\acomp,\dom,\ran,\join)\models s\le t$
and we have to show show that this implies $\axa\vdash s\le t$.
Actually we will show the contrapositive, so we assume that
$\axa\not\vdash s\le t$ and we will construct a representable algebra
$\mathcal{A}\in\repr(\acomp,\dom,\ran,\join)$ witnessing $s\not\le t$:
$\mathcal{A}\not\models s\le t$.

Our assumption that
$\axa\not\vdash s\le t$ is equivalent to
$\freevar\not\models s\le t$.
Instead of $\freevar\models s \le t$ we will sometimes write $\overline{s}\le \overline{t}$.
Let $F^-_\var$ be the set of equivalence classes of join-free terms.

Before embarking on the formal construction of $\mathcal{A}$ we try to give some intuition to the reader.
Let $\mathcal{B}\in\repr(\acomp,\dom,\ran,\join)$, say
$\mathcal{B}\subseteq(\wp(U\times U),\acomp,\dom,\ran,\join)$ for some set $U$.
$\mathcal{B}$ can be viewed as a labelled, directed graph $G=(U,E,\ell)$, where
\begin{itemize}
\item
$U$ is the set of nodes,
\item
$E\subseteq U\times U$ is a set of directed edges,
\item
$\ell\colon U\times U\to\wp(B)$ is a labelling of edges.
\end{itemize}
Indeed, we can define 
$E=\{(u,v)\in U\times U\mid (u,v)\in b\text{ for some }b\in B\}$
and
$\ell(u,v)=\{b\in B\mid (u,v)\in b\}$.
Conversely, given an abstract algebra $\mathcal{B}$ and 
a labelled, directed graph $G=(U,E,\ell)$ satisfying certain conditions (see below),
we can define a representation of $\mathcal{B}$
as a subalgebra of $(\wp(U\times U),\acomp,\dom,\ran,\join)$ by
$$
(u,v)\in b\text{ iff }b\in\ell(u,v)
$$
for every $(u,v)\in U\times U$ and $b\in B$.
The conditions on $G$ must ensure that the operations are correctly represented.
For instance, for angelic composition we need:
\begin{description}
\item[Comp]
$a\acomp b\in\ell(u,v)$ iff there is $w$ such that 
$a\in\ell(u,w)\text{ and }b\in\ell(w,v)$
\end{description} 
for all $a,b\in B$ and $u,v\in U$.
There will be similar conditions for the other operations, see below.
We will use this approach below in constructing the witnessing algebra $\mathcal{A}$
by defining a labelled, directed graph $G_\omega$ where the labels are coming
from the free algebra.

The labelled, directed graph $G_\omega$ will be defined as the union of a chain of labelled, directed graphs 
$G_n= (U_n,E_n,\ell_n)$ for $n\in\omega$, where 
\begin{itemize}
\item
$U_n$ is the set of nodes,
\item
$\ell_n\colon U_n\times U_n\to\wp(F^-_\var)$ is a labelling of edges,
\item
$E_n\subseteq U_n\times U_n$ is a set of edges with non-empty labels.
\end{itemize}
We will make sure that the following \emph{coherence conditions} are maintained during the construction:
\begin{description}
\item[GenC]
$E_n$ is a transitive, reflexive and antisymmetric relation on $U_n$.
\item[PriC]
For every $(u,v)\in E_n$, $\ell_n(u,v)$ is a principal upset: 
$\ell_n(u,v)=a^\uparrow=\{x\in F^-_\var \mid a\le x\}$ 
for some $a\in F^-_\var$.
\item[CompC]
For all $(u,v),(u,w),(w,v)\in U_n\times U_n$ and $a,b\in F^-_\var$, 
if $a\in\ell_n(u,w)$ and $b\in\ell_n(w,v)$, then $a\acomp b\in\ell_n(u,v)$.
\item[DomC]
For all $(u,v)\in U_n\times U_n$ and $a\in F^-_\var$, 
if $\ell_n(u,v)=a^\uparrow$, then $\ell_n(u,u)=\dom(a)^\uparrow$.
\item[RanC]
For all $(u,v)\in U_n\times U_n$ and $a\in F^-_\var$, 
if $\ell_n(u,v)=a^\uparrow$, then $\ell_n(v,v)=\ran(a)^\uparrow$.
\item[IdeC]
For all $(u,v)\in U_n\times U_n$,
$u=v$ iff  $\ell_n(u,v)=\dom(a)^\uparrow$
for some $a\in F^-_\var$.
\end{description}
The construction will terminate in $\omega$ steps, 
yielding $G_\omega=(U_\omega,E_\omega,\ell_\omega)$
where $U_\omega=\bigcup_n U_n$, $\ell_\omega=\bigcup_n \ell_n$,
$E_\omega=\bigcup_n E_n$.

Observe that GenC ensures that $E_\omega$ is antisymmetric,
corresponding to the fact that the free algebra $\freevar$ is free from cycles.
Condition PriC makes sure that the labels are upward closed.
We can view the rest of the coherence conditions as rules
ensuring ``soundness'' of the representation of the operations. 
For instance, CompC tells us that
$a\acomp b\in \ell_\omega(u,v)$ whenever we have a $w$ with
$a\in\ell_\omega(u,w)$ and $b\in\ell_\omega(w,v)$.
But we need the other, ``completeness'', direction as well (see condition Comp above),
that is why we state the saturation conditions below.
By the end of the construction we will achieve the following 
\emph{saturation conditions}:
\begin{description}
\item[CompS]
For all $(u,v)\in U_\omega\times U_\omega$ and $a,b\in F^-_\var$, 
if 
$a\acomp b\in\ell_\omega(u,v)$,
then $a\in\ell_\omega(u,w)$ and $b\in\ell_\omega(w,v)$ for some $w\in U_\omega$.
\item[DomS]
For all $(u,u)\in U_\omega\times U_\omega$ and $a\in F^-_\var$, 
if $\dom(a)\in\ell_\omega(u,u)$, then $a\in\ell_\omega(u,w)$ for some $w\in U_\omega$.
\item[RanS]
For all $(u,u)\in U_\omega\times U_\omega$ and $a\in F^-_\var$, 
if $\ran(a)\in\ell_\omega(u,u)$, then $a\in\ell_\omega(w,u)$ for some $w\in U_\omega$.
\end{description}

Let $\Sigma$ be a fair scheduling function 
$\Sigma\colon\omega\to 3\times\omega\times\omega\times F^-_\var\times F^-_\var$,
i.e., every element of $3\times\omega\times\omega\times F^-_\var\times F^-_\var$
appears infinitely often in the range of $\Sigma$.
The role of $\Sigma$ is to ensure that we deal with every potential ``defect''.

\begin{description}
\item[Initial step]
In the 0th step of the step-by-step construction we define 
$G_0=(U_0,E_0,\ell_0)$
by creating an edge for every element of $F^-_\var$.
We define $U_0$ by choosing  elements $u_a,v_a,\ldots\in\omega$  so that 
$\{u_a,v_a\}\cap\{u_b,v_b\}=\emptyset$  for distinct $a,b$ in $F^-_\var$, and
$u_a=v_a$ iff $\dom(a)=a$ 
(i.e., $a$ is a domain element of $\freevar$).
We can assume that $|\omega\setminus U_0|=\omega$.
We define
\begin{align*}
\ell_0(u_a,v_a)&=a^\uparrow\\
\ell_0(u_a,u_a)&=\dom(a)^\uparrow\\
\ell_0(v_a,v_a)&=\ran(a)^\uparrow
\end{align*} 
and we label all other edges by $\emptyset$.
Thus $E_0=\{(u_a,u_a), (u_a,v_a),(v_a,v_a)\mid a\in F^-_\var\}$.
\end{description}
Observe that the labels are well defined: when $u_a=v_a$ then
$a$ is a domain element, i.e., $\dom(a)=a=\ran(a)$. 

\begin{lemma}\label{lemma-base}
$G_0$ is coherent.
\end{lemma}

\begin{proof}
Conditions GenC and PriC are obvious.
For CompC note that $\dom(a)\acomp\dom(a)=\dom(a)$,
$\dom(a)\acomp a= a = a\acomp\ran(a)$ and
$\ran(a)\acomp\ran(a)=\ran(a)$.
DomC and RanC are straightforward by the definition of the labels.
IdeC follows by the choice of $u_a$ and $v_a$.
\end{proof}

For the successor step $m+1$ we assume inductively that a finite, coherent graph $G_m$ has been constructed.
Let $\Sigma(m+1)=(i,u,v,a,b)$. We will have three types of successor steps depending on the value of $i$.
In each case we will assume that $(u,v)\in E_m$ --- otherwise we define $G_{m+1}=G_m$.
We can assume, by the induction hypothesis PriC, that $\ell_m(u,v)=c^\uparrow$ for some $c\in F^-_\var$.

\begin{description}
\item[Successor step when $\boldsymbol{i=0}$]
Our aim is to extend $G_m$ to create an edge $(u,w)$ witnessing $a$, provided $\dom(a)\in\ell_m(u,v)$.
Thus we assume that $c\le \dom(a)$ --- otherwise we define $G_{m+1}=G_{m}$.
Observe that $c$ must be a domain element (by Claim~\ref{claim-domain2}), i.e., $\dom(c)=c$.
Thus, by IdeC for $G_m$, we have that $u=v$.

We also assume that $\dom(c)\acomp a$ is not a domain element --- 
otherwise we define $G_{m+1}=G_{m}$.
Indeed, if $\dom(c)\acomp a=\dom(\dom(c)\acomp a)$ then using~\eqref{eq-ddcdcd}
we get
\begin{align*}
\dom(c)&=\dom(c)\acomp\dom(a)\\
&=\dom(\dom(c)\acomp a)\\
&=\dom(c)\acomp a\\
&\le a
\end{align*}
whence $c=\dom(c)\le a$, i.e., $a\in\ell_m(u,u)$.

Thus, by the above assumptions,
we have a loop $(u,u)$ labelled by the upset of a domain element $c=\dom(c)\le a$ such that
$\dom(c)\acomp a$ is not a domain element, but we may miss an
edge $(u,w)$ witnessing $a$.

We choose $w\in\omega\setminus U_m$, extend $\ell_m$ by
\begin{align*}
\ell_{m+1}(u,w)&=(\dom(c)\acomp a)^\uparrow\\
\ell_{m+1}(w,w)&=(\ran(\dom(c)\acomp a))^\uparrow\\
\intertext{and for every $(p,u)\in E_m$ with $\ell_m(p,u)=d^\uparrow$ (some $d\in F^-_\var$)}
\ell_{m+1}(p,w)&=(d\acomp a)^\uparrow
\end{align*} 
and all other edges involving the node $w$ have empty labels.
We define $E_{m+1}$ as the set of edges with non-empty labels:
$$
E_{m+1}=\{(u,v)\in U_{m+1}\times U_{m+1}\mid \ell_{m+1}(u,v)\neq\emptyset\}
$$
See Figure~\ref{fig-dom}, where we show the elements whose upsets provide the labels
for the edges.
\end{description}

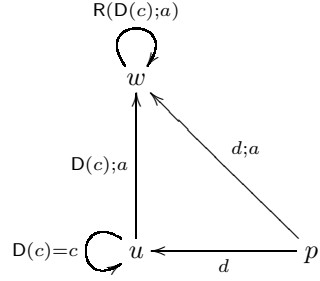
\begin{figure}[h]
\[
\xymatrix{
w\ar@(ul,ur)^{\ran(\dom(c)\acomp a)}
&&\\
&&\\
u\ar@(ul,dl)_{\dom(c)=c}\ar[uu]^{\dom(c)\acomp a}
&&p\ar[ll]^{d}\ar[uull]_{d\acomp a}
}
\]
\caption{Step for domain}\label{fig-dom}
\end{figure}

\begin{description}
\item[Successor step when $\boldsymbol{i=1}$]
This is the mirror image of the previous step for range.
We give an outline, since the reader should not have any difficulty in working out the details.
Our assumption is that
we have a loop $(u,u)$ labelled by the upset of a range element $c=\ran(c)\le a$ such that
$ a\acomp\ran(c)$ is not a range element, but we may miss an
edge $(w,u)$ witnessing $a$.

We choose $w\in\omega\setminus U_m$, extend $\ell_m$ by
\begin{align*}
\ell_{m+1}(w,u)&=(a\acomp\ran(c))^\uparrow\\
\ell_{m+1}(w,w)&=(\dom(a\acomp\ran(c)))^\uparrow\\
\intertext{and for every $(u,p)\in E_m$ with $\ell_m(u,p)=d^\uparrow$ (some $d\in F^-_\var$)}
\ell_{m+1}(w,p)&=(a\acomp d)^\uparrow
\end{align*} 
and all other edges involving the node $w$ have empty labels.
We define $E_{m+1}$ as the set of edges with non-empty labels.
See Figure~\ref{fig-ran}, where we show the elements whose upsets provide the labels
for the edges.
\end{description}

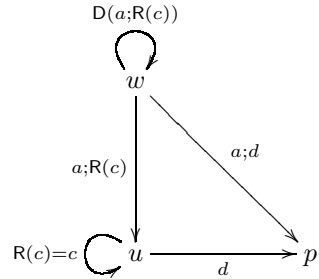
\begin{figure}[h]
\[
\xymatrix{
w\ar@(ul,ur)^{\dom(a\acomp\ran(c))}\ar[dd]_{a\acomp\ran(c)}\ar[ddrr]^{a\acomp d}
&&\\
&&\\
u\ar@(ul,dl)_{\ran(c)=c}\ar[rr]_{d}
&&p
}
\]
\caption{Step for range}\label{fig-ran}
\end{figure}

\begin{description}
\item[Successor step when $\boldsymbol{i=2}$]
In this case our aim is to extend $G_m$ to create edges $(u,w)$ and $(w,v)$ witnessing $a$ and $b$,
provided $a\acomp b\in\ell_m(u,v)$.
Thus we assume that $c\le a\acomp b$ --- otherwise we define $G_{m+1}=G_{m}$.

We can also assume that $u\ne v$ because of the following.
If $u=v$, then, by IdeC for $G_m$, the element $c$ must be a domain element $c=\dom(c)$.
Thus we have $c=\dom(c)\le a\acomp b$.
Using Claim~\ref{claim-refl} we get that both $a$ and $b$ are domain elements
and that $c=\dom(c)\le a$ and $c=\dom(c)\le b$,
whence $a,b\in\ell_m(u,u)$.

Finally, we can assume that neither
$\dom(c)\acomp a\acomp \dom(b\acomp\ran(c))$ nor
$\ran(\dom(c)\acomp a)\acomp b\acomp\ran(c)$ is a domain element because of the following.
Assume that $d=\dom(c)\acomp a\acomp \dom(b\acomp\ran(c))$ is a domain element, $d=\dom(d)$.
Recall that $c\le a\acomp b$, whence
\begin{equation*}
c=\dom(c)\acomp c\acomp\ran(c)\le\dom(c)\acomp a\acomp b\acomp \ran(c)
\end{equation*}
so
\begin{align*}
\dom(c)&\le\dom(\dom(c)\acomp a\acomp b\acomp \ran(c))&&\\
&=\dom(\dom(c)\acomp a\acomp\dom(b\acomp \ran(c)))&&\text{by~\eqref{eq-dcdcd}}\\
&=\dom(c)\acomp a\acomp \dom(b\acomp \ran(c))&&\text{by~$d=\dom(d)$}\\
&\le a
\intertext{i.e., $a\in\ell_m(u,u)$.
Also,}
c&\le\dom(c)\acomp a\acomp b\acomp \ran(c)&&\\
&=\dom(c)\acomp a\acomp \dom(b\acomp \ran(c))\acomp b\acomp \ran(c)&&\text{by~\eqref{eq-d}}\\
&=d\acomp b\acomp \ran(c)&&\\
&\le b&&\text{by~$d=\dom(d)$}
\end{align*}
i.e., $b\in\ell_m(u,v)$.
Thus $a \acomp b\in \ell_m(u,v)$ by CompC.

Showing that the required witness edges already exist in $G_m$ when
$\ran(\dom(c)\acomp a)\acomp b\acomp\ran(c)$ is a domain element
is completely analogous.

Summing up: we assume that
\begin{itemize}
\item[(CC1)]
$c\le a\acomp b$
\item[(CC2)]
$u\ne v$
\item[(CC3)]
$\dom(c)\acomp a\acomp \dom(b\acomp\ran(c))\ne \dom(\dom(c)\acomp a\acomp \dom(b\acomp\ran(c)))$
\item[(CC4)] 
$\ran(\dom(c)\acomp a)\acomp b\acomp\ran(c)\ne \ran(\ran(\dom(c)\acomp a)\acomp b\acomp\ran(c))$
\end{itemize}
otherwise we define $G_{m+1}=G_m$.
If (CC1)--(CC4) hold, then we choose $w\in\omega\setminus U_m$, extend $\ell_m$ by
\begin{align*}
\ell_{m+1}(u,w)&=(\dom(c)\acomp a\acomp \dom(b\acomp\ran(c)))^\uparrow\\
\ell_{m+1}(w,v)&=(\ran(\dom(c)\acomp a)\acomp b\acomp\ran(c))^\uparrow\\
\ell_{m+1}(w,w)&=(\ran(\dom(c)\acomp a)\acomp \dom(b\acomp\ran(c)))^\uparrow\\
\intertext{and for $(p,u),(v,q)\in E_m$ with $\ell_m(p,u)=d^\uparrow$ and
$\ell_m(v,q)=e^\uparrow$ (some $d,e\in F^-_\var$)}
\ell_{m+1}(p,w)&=(d\acomp  a\acomp \dom(b\acomp\ran(c)))^\uparrow\\
\ell_{m+1}(w,q)&=(\ran(\dom(c)\acomp a)\acomp b\acomp e)^\uparrow
\end{align*} 
and all other edges involving $w$ will have empty labels.
We define $E_{m+1}$ as the set of edges with non-empty labels:
See Figure~\ref{fig-comp}, where the label-generating elements are indicated.
\end{description}

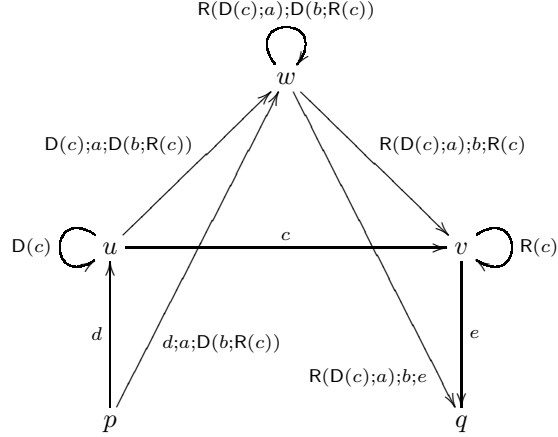
\begin{figure}[h]
\[
\xymatrix{
&&w\ar@(ul,ur)^{\ran(\dom(c)\acomp a)\acomp \dom(b\acomp\ran(c))}\ar[ddrr]^{\ran(\dom(c)\acomp a)\acomp b\acomp\ran(c)}
\ar[ddddrr]_(.85){\ran(\dom(c)\acomp a)\acomp b\acomp e}&&\\
&&&&\\
u\ar@(ul,dl)_{\dom(c)}\ar[uurr]^{\dom(c)\acomp a\acomp \dom(b\acomp\ran(c))}
\ar[rrrr]^{c}&&&&
v\ar@(ur,dr)^{\ran(c)}\ar[dd]^{e}\\
&&&&\\
p\ar[uu]^{d}
\ar[uuuurr]_(.25){d\acomp  a\acomp \dom(b\acomp\ran(c))}&&&&
q
}
\]
\caption{Step for composition}\label{fig-comp}
\end{figure}

\begin{lemma}\label{lem-coh}
Every $G_{m+1}$ is coherent.
\end{lemma}

\begin{proof}
First assume that  $i=0$ and $G_m\ne G_{m+1}$.
Let $p\ne u$ and $d$ be such that $\ell_{m}(p,u)=d^\uparrow$
so that $\ell_{m+1}(u,w)=(d\acomp a)^\uparrow$,
see Figure~\ref{fig-dom}.

Conditions  GenC and PriC are obvious. 
CompC easily follows by the definition of the labels on the new edges,
we just elaborate one case: the triangle consisting of the edges $(p,u)$, $(u,w)$ and $(p,w)$.
By RanC for $G_m$ we have $\dom(c) =\ran( d)$,
thus $(d\acomp\dom(c)\acomp a)^\uparrow = (d\acomp a)^\uparrow$,
as desired.

For DomC we need first
$\dom(c)=\dom(\dom(c))=\dom(\dom(c)\acomp\dom(a))=\dom(\dom(c)\acomp a)$,
by $\dom(c)\le\dom(a)$ and~\eqref{eq-dcdcd}.
Next we show
$\dom(d)=\dom(d\acomp a)$.
Observe that $\dom(c)=\ran(d)$ by RanC for $G_m$, whence 
$\dom(d\acomp a)=\dom(d\acomp\dom(c)\acomp a)= \dom(d\acomp\dom(\dom(c)\acomp a))
=\dom(d\acomp\dom(c))=\dom(d)$ 
by~\eqref{eq-dcdcd}.
A similar argument using~\eqref{eq-rcrrc} shows that $\ran(\dom(c)\acomp a)=\ran(d\acomp a)$, 
establishing RanC.
Finally, for IdeC we need that $\dom(c)\acomp a$ is not a domain element
(which is obvious by our assumption we made during the construction)
and similarly for $d\acomp a$.
By Claim~\ref{claim-refl}, 
$d\acomp a\in\dom(F_\var)$ would imply $d, a\in\dom(F_\var)$. 
But $d\notin\dom(F_\var)$ by IdeC for $G_m$, 
and $a\notin\dom(F_\var)$, since this would imply $\dom(c)\acomp a\in\dom(F_\var)$.

The case for  $i=1$ is completely analogous.

Finally assume that $i=2$ and $G_m\ne G_{m+1}$.
Let $p,q\notin\{u,v\}$ and $d,e$ be such that
$\ell_m(p,u)=d^\uparrow$ and $\ell_m(v,q)=e^\uparrow$,
see Figure~\ref{fig-comp}.

Conditions  GenC and PriC are obvious again.
Next consider DomC and RanC.
First we look at the edge $(u,w)$.
We need
\begin{equation}\label{eq-domc}
\dom(c)=\dom(\dom(c)\acomp a\acomp \dom(b\acomp\ran(c)))
\end{equation}
and
$\ran(\dom(c)\acomp a)\acomp \dom(b\acomp\ran(c))=\ran(\dom(c)\acomp a\acomp \dom(b\acomp\ran(c)))$.  Now
$\dom(c)\ge \dom(\dom(c)\acomp a\acomp \dom(b\acomp\ran(c)))$ is obvious.
For the other direction recall that we assumed that $c\le a\acomp b$,
whence
\begin{align*}
\dom(c)&=\dom(\dom(c)\acomp c\acomp \ran(c))\\
&\le \dom(\dom(c)\acomp a\acomp b\acomp \ran(c))\\
&=\dom(\dom(c)\acomp a\acomp \dom(b\acomp \ran(c)))
\end{align*}
by~\eqref{eq-dcdcd}, and using~\eqref{eq-rd} and \eqref{eq-rcrrcr} we get
\begin{align*}
\ran(\dom(c)\acomp a)\acomp \dom(b\acomp\ran(c))&=
\ran(\dom(c)\acomp a)\acomp \ran(\dom(b\acomp\ran(c)))\\
&=\ran(\dom(c)\acomp a\acomp \ran(\dom(b\acomp\ran(c))))\\
&=\ran(\dom(c)\acomp a\acomp \dom(b\acomp\ran(c)))
\end{align*}
as desired.
Next we consider the edge $(p,w)$.
Recall that $\dom(c)=\ran(d)$ by RanC for $G_m$.
Then we get
\begin{align*}
\dom(d)&=\dom(d\acomp\dom(c))&&\\
&=\dom(d\acomp\dom(\dom(c)\acomp a\acomp \dom(b\acomp\ran(c))))&&\text{by \eqref{eq-domc}}\\
&=\dom(d\acomp\dom(c)\acomp a\acomp \dom(b\acomp\ran(c)))&&\text{by \eqref{eq-dcdcd}}\\
&=\dom(d\acomp a\acomp \dom(b\acomp\ran(c)))&&\\
\end{align*}
and
\begin{align*}
\ran(\dom(c)\acomp a)\acomp\dom(b\acomp\ran(c))&=
\ran(\dom(c)\acomp a)\acomp\ran(\dom(b\acomp\ran(c)))&&\text{by \eqref{eq-rd}}\\
&=\ran(\ran(d)\acomp a)\acomp\ran(\dom(b\acomp\ran(c)))&&\\
&=\ran(d\acomp a)\acomp\ran(\dom(b\acomp\ran(c)))&&\text{by \eqref{eq-rcrrc}}\\
&=\ran(d\acomp a\acomp\ran(\dom(b\acomp\ran(c))))&&\text{by \eqref{eq-rcrrcr}}\\
&=\ran(d\acomp a\acomp\dom(b\acomp\ran(c)))&&\text{by \eqref{eq-rd}}\\
\end{align*}
as desired.
Checking DomC and RanC for other edges is completely analogous.

For CompC the main observation is the following.
By our assumption $c\le a\acomp b$,
whence
\begin{align*}
c&\le\dom(c)\acomp a\acomp b\acomp\ran(c)\\
&=\dom(c)\acomp a\acomp\ran(\dom(c)\acomp a)\acomp \dom(b\acomp\ran(c))\acomp b\acomp\ran(c)\\
&=\dom(c)\acomp a\acomp \dom(b\acomp\ran(c))\acomp\ran(\dom(c)\acomp a)\acomp b\acomp\ran(c)
\end{align*} 
thus CompC holds for the triangle consisting of the edges $(u,w)$, $(w,v)$ and $(u,v)$.
Checking CompC for the other triangles is easy, using the definition of the labels
and the already established DomC and RanC.

Finally, checking IdeC can be done similarly to the case $i=0$,
by using the assumption that $(u,w)$ and $(w,v)$ are labelled by non-domain elements
and that, for every $x,y\in F^-_\var$, we have
$x\acomp y\in \dom(F_\var)$ iff both $x$ and $y$ are in $\dom(F_\var)$.
\end{proof}

\begin{lemma}\label{lem-sat}
$G_\omega$ is coherent and saturated.
\end{lemma}

\begin{proof}
Coherence of $G_\omega$ follows from the coherence of all $G_m$.

Let us check DomS. 
Let $(u,u)\in U_\omega\times U_\omega$  be such that $\dom(a)\in\ell_\omega(u,u)$.
By coherence of $G_\omega$ we have that $\ell_\omega(u,u)=\dom(c)^\uparrow$
for some element $c$.
If $\dom(c)\acomp a$ is a domain element, then $a\in \ell_\omega(u,u)$,
as we have seen in case $i=0$ of the successor step of the construction.
If $\dom(c)\acomp a\notin \dom(F_\var)$, consider $m\in\omega$ such that 
$\Sigma(m+1)=(0,u,u,a,a)$ and $\dom(a)\in\ell_m(u,u)$.
Then $G_{m+1}$ contains an edge $(u,w)$ such that $\ell_{m+1}(u,w)=(\dom(c)\acomp a)^\uparrow$.
By \eqref{eq-dcord}
we get $\dom(c)\acomp a\le a$, i.e.,
$a\in\ell_{m+1}(u,w)=\ell_\omega(u,w)$.
Checking RanS is completely analogous.

For CompS assume that $a\acomp b\in\ell_\omega(u,v)=c^\uparrow$, 
and let $m$ be such that $\Sigma(m+1)=(2,u,v,a,b)$ and $a\acomp b\in\ell_m(u,v)$.
Recall that we showed at case $i=2$ of the successor step of the construction
that the edges witnessing $a$ and $b$ already exist in $G_m$ if any of the conditions
(CC2)--(CC3) fails.
If the conditions hold, then $G_{m+1}$ contains $w$ such that 
$(\dom(c)\acomp a\acomp\dom(b\acomp\ran(c)))^\uparrow=\ell_{m+1}(u,w)=\ell_\omega(u,w)$ and 
$(\ran(\dom(c)\acomp a)\acomp b\acomp\ran(c))^\uparrow=\ell_{m+1}(w,v)=\ell_\omega(w,v)$.
Using Claim~\ref{claim-sl} we get $a\in\ell_\omega(u,w)$ and $b\in\ell_\omega(w,v)$
as desired.
\end{proof}

Next we define a valuation $\flat$ of variables.
Recall that for term $r$, its equivalence class in $F_\var$ is denoted by $\overline{r}$.
We let
\begin{equation*}
x^\flat=\{(u,v)\in U_\omega\times U_\omega: \overline{x}\in\ell_\omega(u,v)\}
\end{equation*}
for every variable $x\in \var$.
Let $\mathcal{A}=(A,\acomp,\dom,\ran,\join)$ be the subalgebra
of the full algebra $(\wp(U_\omega\times U_\omega),\acomp,\dom,\ran,\join)$ generated by
$\{x^\flat:x\in \var\}$.
Clearly $\mathcal{A}$ is representable.

\begin{lemma}~\label{lem-true}
For every  join-free term $r$ and $(u,v)\in U_\omega\times U_\omega$,
\begin{equation*}
(u,v)\in r^\flat\mbox{ iff }
\overline{r}\in\ell_\omega(u,v)
\end{equation*}
where $r^\flat$ is the interpretation of
$r$ in $\mathcal{A}$ under the valuation $\flat$.
\end{lemma}

\begin{proof}
This is an easy induction on the complexity of terms, using that $\mathcal{A}$ is representable.
For the left-to-right direction use that $G_\omega$ satisfies the coherence conditions
CompC, DomC and RanC.
For the right-to-left direction use 
that $G_\omega$ satisfies the saturation conditions
CompS, DomS and RanS.
\end{proof}

Recall that we assumed that $\mathcal{F}_\var\not\models s\le t$.
In the initial step of the construction we created the edge
$(u_{\overline{s}}, v_{\overline{s}})$ such that 
$\ell_0(u_{\overline{s}},v_{\overline{s}})=\overline{s}^\uparrow$.
Thus $\overline{s}\in\ell_\omega(u_{\overline{s}},v_{\overline{s}})$ and 
$\overline{t}\notin\ell_\omega(u_{\overline{s}},v_{\overline{s}})$.
Hence, by Lemma~\ref{lem-true}, $(u_{\overline{s}},v_{\overline{s}})\in s^\flat$ and 
$(u_{\overline{s}},v_{\overline{s}})\notin t^\flat$.
That is, $\mathcal{A}\not\models s\le t$, as desired.  This completes the final  step required in the proof of Theorem \ref{thm-main}.

\section{Demonic composition}
In this section we look at demonic composition.
We define the following set $\axd$ of axioms:
\begin{align}
x\dcomp (y\dcomp z)&=(x\dcomp y)\dcomp z\label{eqn:associativity}\\
\D(x)\dcomp x&=x\label{eqn:leftid}\\
\D(x)\dcomp \D(y)&=\D(y)\dcomp \D(x)\label{eqn:comm}\\
\D(\D(x)\dcomp y)&=\D(x)\dcomp \D(y)\label{eqn:normal}\\
x\dcomp \D(y)&=\D(x\dcomp y)\dcomp x\label{eqn:Dtwisted}\\
 \D\R(x)&=\R(x)\label{eqn:DR}\\
 \R\D(x)&=\D(x)\label{eqn:RD}\\
 \R\R(x)&=\R(x)\label{eqn:RR}\\
 \R(x)\dcomp \R(y)&=\R(y)\dcomp \R(x)\label{eqn:Rcomm}\\
 \R(x\dcomp y)\dcomp \R(y)&=\R(x\dcomp y)\label{eqn:Rorder}\\
 x\dcomp \R(x)&=x\label{eqn:rightid}
 \end{align}
 These axioms differ from \eqref{eq-ass}--\eqref{eq-rcom} by the inclusion of the stronger domain property \eqref{eqn:Dtwisted} and the weakening of the range property \eqref{eq-rcrrc} to \eqref{eqn:Rorder}.
Axioms \eqref{eqn:associativity}--\eqref{eqn:Dtwisted} state that 
the $\R$-free reduct of an algebra $\mathcal{S}=( S, \dcomp , \D,\R)$ is a restriction semigroup, 
while the remaining laws make $\mathcal{S}$ a two-sided closure semigroup in the sense of 
Jackson and Stokes~\cite{jacsto:01}.  

We claim that $\axd$ is sound for demonic composition of binary relations with domain and range,
i.e.,  for the representation class $\repr(\dcomp,\dom,\ran)$.
Soundness of the laws \eqref{eqn:associativity}--\eqref{eqn:Dtwisted} is observed in 
Desharnais, Jipsen and Struth~\cite{DJS}.  
Laws \eqref{eqn:DR}--\eqref{eqn:rightid} simply ensure that domain and range elements coincide, 
and that $\R(x)$ is the smallest domain element (with respect to the usual order in a meet semilattice) 
in the abstract algebra that acts as a right identity for~$x$; see~\cite[Proposition~1.2]{jacsto:01}.

Observe that \eqref{eqn:Dtwisted} fails for angelic composition.
Indeed, we may have
$(a,b)\in \D(x\acomp y)\acomp x$ (because $(a,b)\in x$ and we have some $c,d$ such that
$(a,c)\in x$ and $(c,d)\in y$)
while $(a,b)\notin x\acomp \D(y)$ (because we may have $(b,b)\notin \D(y)$).

The main result of this section is that $\axd$ provides a complete equational axiomatization 
for demonic composition with domain and range.

\begin{theorem}\label{thm-maind}
The variety $\vari(\dcomp,\dom,\ran)$ generated by the representation class $\repr(\dcomp,\dom,\ran)$
is axiomatized by $\axd$.
\end{theorem}

\begin{proof}
Let us say that a model of $\axd$ is \emph{cycle free}, if it satisfies the laws 
\begin{align}
x\dcomp y=x&\Rightarrow \D(y)=y\label{eqn:rightdegenerate}\\
x\dcomp y= \D(z)&\Rightarrow x=\D(x)\And y=\D(y)\label{eqn:domainprime}
\end{align}
for every $x,y,z$.
We will show in Lemma~\ref{lem-freecf} that the free algebras of the variety defined by $\axd$ are cycle free (this is somewhat analogous to Claims \ref{claim-domain}--\ref{claim-domain2} of the previous section).
Furthermore, by Lemma~\ref{lem-repcf}, cycle-free models of $\axd$ are in the variety 
$\vari(\dcomp,\dom,\ran)$ (in fact, they are representable).
\end{proof}
The rest of this section is devoted to prove these lemmas.

\subsection{The free algebra}


Let $s$ be a term in the variables $x_1,\dots,x_n$, and let $\operatorname{out}(s)$ denote 
the set of pairs $(x_i,j)$, where $j\in\omega$ denotes the number of occurrences of the variable $x_i$ 
that lie outside of any application of $\D$ or $\R$.  
Let $\out_{x_i}(s)$ denote the value $j$ such that $(x_i,j)\in \out(s)$.

\begin{lemma}\label{lem:freeimplication}
If  $\axd\vdash s=t$, then $\operatorname{out}(s)=\operatorname{out}(t)$,
for any pair $s,t$ of terms.
\end{lemma}

\begin{proof}
This follows by induction on the length of a derivation: 
we show that the statement holds true for each individual step of deduction, 
and hence by a trivial induction argument, also for any finite sequence of deductions.  
A step of deduction starting from $s$ involves replacing a subterm of $s$ 
by the image of an axiom under some substitution.  
Assume we have an axiom $u=v$ (or $v=u$) from~$\axd$, a substitution~$\theta$, 
and are replacing an instance of the subterm $\theta(u)$ of $s$ by $\theta(v)$ to obtain a term $s'$.  
Thus it suffices to show that $\out(\theta(u))=\out(\theta(v))$.  
This follows because $\out(u)=\out(v)$, as can be seen by inspection of the laws in $\axd$.
\end{proof}

\begin{lemma}\label{lem-freecf}
The free algebras of the variety defined by $\axd$ are cycle free.
\end{lemma}

\begin{proof}
For the first law \eqref{eqn:rightdegenerate}, assume that $s$ and $t$ are terms such that $s\dcomp t=s$.  
We show that $t=\D(t)$ is a consequence of $\axd$.  
Now, as $\axd\vdash s\dcomp t=s$, Lemma~\ref{lem:freeimplication} shows that $\out(s\dcomp t)=\out(s)$.  
As for all variables $x$, we have $\out_x(s\dcomp t)=\out_x(s)+\out_x(t)$ 
it follows that $\out_x(t)=0$ always.  
Thus all variables in $t$ lie under an application of $\D$ or $\R$, showing that $\axd\vdash t=\D(t)$, 
as required.

For the second law \eqref{eqn:domainprime}, assume that $s,t,d$ are terms such that 
$\axd\vdash s\dcomp t=\D(d)$.  An almost identical argument to before shows that all variables 
in $s$ and $t$ have all occurrences under an application of $\D$ or $\R$, 
and hence $\axd\vdash s=\D(s)$ and $t=\D(t)$.
\end{proof}

\subsection{Representing cycle-free algebras}

\begin{lemma}\label{lem-repcf}
Any cycle-free model $\mathcal{S}=(S,\dcomp ,\D,\R)$ of $\axd$
is representable: $\mathcal{S}\in\repr(\dcomp,\dom,\ran)$.
\end{lemma}

\begin{proof}
We piece together a representation $\theta$ by way of an inductive gluing of pieces of 
the Wagner--Preston $(\dcomp,\D )$-representations.  
Recall \cite[Theorem~3.9]{jacsto:01} that a restriction semigroup 
$(S,\dcomp ,\D)$ can be represented as partial maps over itself by
\[
a\mapsto \{(x\dcomp D(a),x\dcomp a)\mid x\in S\}=\{(x,x\dcomp a)\mid x\in S\dcomp \D(a)\}.
\]
For any $s\in S$ we let $F_s$, the \emph{forward closure} of $s$, 
denote the labelled directed graph obtained from this representation on the induced subgraph reachable 
from the point $s$.  At each step of our inductive gluing we will have a $(\dcomp,\D)$-representation 
which does not necessarily correctly represent $\R$.  
A \emph{range defect} for an element $s\in S$ will be a point $p$ of the representation 
in which the domain element $\R(s)$ is defined, but for which $p$ is not in the range of $s^\theta$.

\begin{claim}
If a restriction semigroup $(S,\dcomp,\D)$  is cycle free then the only cycles in the Wagner--Preston representation of $(S,\dcomp,\D)$ are loops.
\end{claim}

\begin{proof}
Assume for contradiction that there is a cycle in the Wagner--Preston representation of $(S,\dcomp,\D)$.  
The underlying graph of this representation is transitive, so there are $x,a,b\in S$ with $x=x\dcomp \D(a)$ 
and $x\dcomp a\dcomp \D(b)=x\dcomp a$ and $x\dcomp a\dcomp b=x$.  
Then by~\eqref{eqn:rightdegenerate} we have $a\dcomp b=\D(a\dcomp b)$.  
By~\eqref{eqn:domainprime} we have $a=\D(a)$ and $b=\D(b)$, 
so that in fact $x=x\dcomp a=x\dcomp a\dcomp b$, and the cycle is a loop.
\end{proof}

\begin{claim}
In a cycle-free restriction semigroup $(S,\dcomp ,\D)$, for any element $s\neq \D(s)$, 
we have $\D(s)\in F_{\D(s)}\backslash F_s$.
\end{claim}

\begin{proof}
Otherwise there would be $b$ with $s\dcomp b=\D(s)$.  Then $s=\D(s)$ by~\eqref{eqn:domainprime}.
\end{proof}

Let $\mathcal{S}=(S,\dcomp ,\D,\R)$ be a cycle-free model of $\axd$.
We will take the union over an $\omega$-chain of partial representations 
$\theta_0,\theta_1,\theta_2,\dots$ over sets $X_0\subseteq X_1\subseteq X_2\subseteq\dots$, 
where $\theta_{i+1}$ coincides with $\theta_i$ when restricted to $X_i$.  
The partial representations are constructed inductively, with the Wagner--Preston representation 
as the base case $\theta_0$ (so that $X_0$ is the universe $S$ of the algebra).  
We have the following inductive hypothesis.
\begin{enumerate}
\item 
Domain and demonic composition are correctly represented by $\theta_i$.
\item 
The partial representation $\theta_i$ is faithful in the sense that for $s\neq t$ there are points $p,q$ 
with $(p,q)\in s^\theta\backslash t^\theta$ or $(p,q)\in t^\theta\backslash s^\theta$.
\end{enumerate}
Note that by Hypothesis (1), only range might fail to be represented properly by $\theta_i$.  
However, as $s\dcomp \R(s)=s$ and $\R(s)$ is a domain element (by $\D(\R(s))=\R(s)$), 
it follows from Hypothesis~(1) that we do at least have $\R(s^{\theta_i})\subseteq \R(s)^{\theta_i}$.  
The construction of $X_{i+1}$ and $\theta_{i+1}$ will be such that all range defects present 
in $X_i$ under $\theta_i$ are corrected by $\theta_{i+1}$, though new range defects may have been
introduced at points in $X_{i+1}\backslash X_i$.  In this way there will be no range defects 
in the final representation over $\bigcup_{i\in\omega}X_i$ so that we will have achieved 
the desired representation of  $\mathcal{S}$.

We now begin the induction.
The conditions hold for the base case: 
the inductive hypothesis simply states that we have a faithful $(\dcomp ,\D)$ representation, 
which holds by~\cite[Theorem~3.9]{jacsto:01}.  

Let us assume that the inductive hypothesis holds on the partial representation $\theta_i$ of 
$(S,\dcomp ,\D,\R)$ over $X_i$.  
If range is correctly represented by $\theta_i$, then our proof is complete: 
we may let $X_{i+1}=X_i$ and $\theta_{i+1}=\theta_i$.  
Otherwise there are range defects in $X_i$ under $\theta_i$.  
We explain how to correct any such range defect.  
The set $X_{i+1}$ and partial representation $\theta_{i+1}$ are obtained by 
simultaneously applying the described method to all range defects in $X_i$ under $\theta_i$.  
To avoid proliferation of indices and symbols, for the remainder of the argument we use 
$X$ to abbreviate $X_i$ and $\theta$ to abbreviate $\theta_i$.  
Let $p\in X$ be a range defect for some element $s\in S$ under $\theta$: 
so $p\pointsto{\R(s)} p$, but $p$ is not in the range of $s^\theta$; 
that is, no point $p'\in X$ has $p'\pointsto{s}p$.  
Note that in this instance it cannot be that $s$ is a domain element, 
as then $s=\R(s)$ which would give $p\pointsto{s}p$.  
Thus $F_{s}$ is a proper subgraph of $F_{\D(s)}$.
Adjoin a disjoint copy of the forward closure $F_{\D(s)}$ to~$X$.  
Retain all edges and labels already existing in $X$ and in the newly adjoined $F_{\D(s)}$, 
but we add some new edges between $F_{\D(s)}\backslash F_{s}$ and $X$.  

Before we describe these new \emph{connector edges}, we observe the following lemma, 
which guarantees that suitable target points in $X$ exist.  
The situation is also depicted in Figure~\ref{fig:copy}, 
where boldface vertices and edges correspond to given assumptions in the lemma, 
and non-boldface vertices and edges are those shown to exist in the lemma.
%

\begin{figure}
\begin{tikzpicture}
\node (F) at (-.5,4,0) {$F_{\D(s)}$:};
\draw[dashed] (1,3) to [out=10,in=140] (10,3);
\draw[dashed] (1,3) to [out=-10,in=220] (10,3);
\node [vertex] (ds) at (1,3) {};
\node at (ds) [above] {$\D(s)$};
\node [vertex] (s) at (3,3) {};
\node at (s) [above] {$s$};
\draw [->,thick] (ds) to (s);
\node at (2.5,3.15)  {$s$};
\node [vertex] (t) at (7,3.4) {};
\node at (t) [above] {$t$};
\draw [->,thick] (s) to (t);
\node at (5.3,3.45)  {$a$};

\draw[dashed] (0,.5) to [out=20,in=160] (10,.5);
\node (X) at (-.5,1) {$X$:};
\node [vertex] (p) at (5,0.5) {};
\node at (p) [below] {$p$};
\draw (p) [out = 135, in = 45, looseness = 15,->,thick] to (p);
\node at (5.2,1) {$\R(s)$};
\node [vertexe] (q) at (7,0.7) {};
\node at (q) [below] {$\exists q$};
\draw [->] (p) to (q);
\node at (6.2,.82) {$\exists a$};
\end{tikzpicture}
\caption{Diagram depicting Lemma \ref{lem:copy}.  
A fresh copy of $F_{\D(s)}$ has been placed aside $X$ in order to eventually correct a defect 
in the range of $s^\theta$ at $p$ in $X$.  
If $s\pointsto{a}t$ then there exists $q$ such that $p\pointsto{a}q$.}\label{fig:copy}
\end{figure}
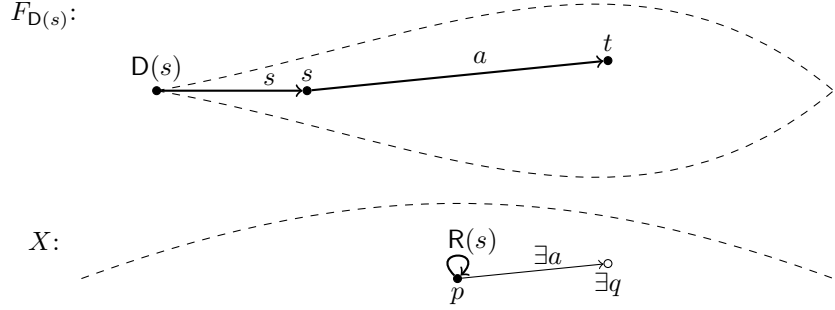

\begin{lemma}\label{lem:copy}
If $s\pointsto{a}t$ then there is $q\in X$ such that $p\pointsto{a}q$.
\end{lemma}

\begin{proof}
If $s\pointsto{a}t$ then $s\dcomp \D(a)=s$ so that $\R(s)\dcomp \D(a)=\R(s)$, 
as $\R(s)$ is the smallest domain element that acts as a right identity for $s$.  
As $\R(s)$ is defined at $p$, we also have $\R(s)\dcomp \D(a)$ defined at $p$, 
and as both $\D(a)$ and $\R(s)$ are domain elements, 
Hypothesis~(1) again ensures that $p\pointsto{\D(a)}p$.  
Then by Hypothesis~(1) again, there is $q\in X$ with $p\pointsto{a}q$.
\end{proof}

Now we describe the connector edges.
Let $v\in F_{s}$ and $u\in F_{\D(s)}$ with $u\neq v$ and assume that  $v=s\dcomp c$ for some $c$.  
For each edge $u\pointsto{a}v$ in $F_{\D(s)}$, Lemma~\ref{lem:copy} shows that 
there exists at least one $q\in X$ with $p\pointsto{c} q$.  
For every such $q$ we add the edge $u\pointsto{a} q$.  
In particular, the edge $\D(s)\pointsto{s}s$ starts at $\D(s)\in F_{\D(s)}$ and ends at $s\in F_s$, 
so that the range defect for $s$ at $p$ is cured.


\begin{figure}
\begin{tikzpicture}
\node (F) at (-.5,4,0) {$F_{\D(s)}$:};
\draw[dashed] (1,3) to [out=10,in=140] (10,3);
\draw[dashed] (1,3) to [out=-10,in=220] (10,3);
\node [vertex] (ds) at (1,3) {};
\node at (ds) [above] {$\D(s)$};
\node [vertex] (u) at (4,3) {};
\node at (u) [above] {$u$};
\node [vertex] (v) at  (6,3) {};
\node at (v) [above] {$v$};
\node [vertex] (w) at (8,3.2) {};
\node at (w) [above] {$w$};
\draw  (u) [->,thick] to (v);
\node at (5,3.15)  {$a$};
\draw (v)  [->,thick] to (w);
\node at (7,3.3)  {$b$};

\draw[dashed] (0,.5) to [out=20,in=160] (10,.5);
\node (X) at (-.5,1) {$X$:};
\node [vertex] (q) at (6,0.7) {};
\node at (q) [below] {$q$};
\node [vertexe] (qq) at (8,0.7) {};
\node at (qq) [below] {$\exists q'$};
\draw  (q) [->] to (qq);
\node at (7.4,0.53)  {$\exists b$};
\draw [->,thick] (u) to (q);
\node at (4.9,1.7)  {$a$};
\draw [->] (u) to (qq);
\node at (6,2.3)  {$\exists a\dcomp b$};
\draw [->] (v) [out=340,in=100] to (qq);
\node at (7.55,2.3)  {$\exists b$};
\end{tikzpicture}
\caption{Diagram depicting Lemma \ref{lem:copy2}.  If $u\pointsto{a}v\pointsto{b}w$ in $F_{\D(s)}$ and $u\pointsto{a}q$ is a connector edge, then there exists $q'$ in $X$ with $q\pointsto{b} q'$ and such that $u\pointsto{a\dcomp b}q'$ is a connector edge.}\label{fig:copy2}
\end{figure}
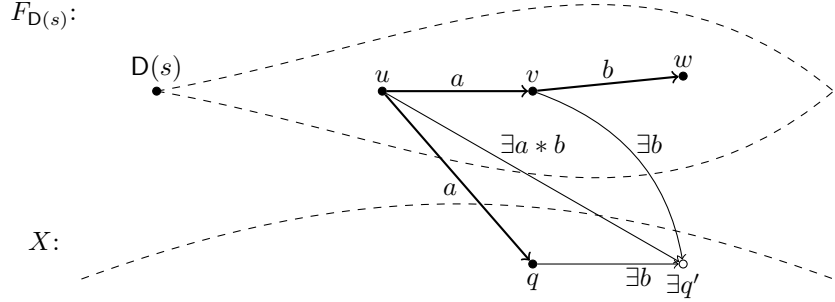

Before we verify the inductive hypotheses are maintained, we observe the following useful lemma, 
which is also depicted in Figure~\ref{fig:copy2} (in the case when $b\neq \D(b)$), 
with the same convention on boldface edges and vertices as in Figure~\ref{fig:copy}.

\begin{lemma}\label{lem:copy2}
Let $u\pointsto{a}q$ be a connector edge \up(so $u\in F_{\D(s)}$ and $q\in X$\up), 
and $v\in F_{\D(s)}$ have $u\pointsto{a}v$.  If $v\pointsto{b}w$ in $F_{\D(s)}$, 
then there is $q'\in X$ with $q\pointsto{b}q'$ in $X$ and for every such $q'$\up:
\begin{itemize}
\item  
$u\pointsto{a\dcomp b}q'$
\item 
if $b$ is a domain element then $q=q'$ and $v=w$, but otherwise $v\pointsto{b}q'$.
\end{itemize}
\end{lemma}

\begin{proof}
Assume the hypothesis of the lemma.  
Because $u\pointsto{a}q$ is a connector edge, there must exist some $v'\in S$ and $c\in S$ 
with $s\pointsto{c}v'$, and $u\pointsto{a}v'$ as well as $p\pointsto{c}q$.    
Because all elements acted as functions within $F_{\D(s)}$ and both $u\pointsto{a}v$ and $u\pointsto{a}v'$, 
we must have $v=v'$.  Indeed, we can write $v$ as $s\dcomp c$ and then $w$ as $s\dcomp c\dcomp b$, 
so that $s\pointsto{c\dcomp b}w$.   Also, as $u\pointsto{a}v\pointsto{b}w$ and 
because demonic composition coincides with angelic composition within $F_{\D(s)}$ 
we have $u\pointsto{a\dcomp b}w$.  
Hence, by Lemma~\ref{lem:copy} we have that $a\dcomp b$ is defined at $p$ in $X$.  
By Hypothesis~(1), every edge labelled by $a$ leaving $p$ is in the domain of $b$.  
In particular, $b$ is defined at $q$, so that there exists points $q'$ such that $p\pointsto{b}q'$. 
Because $s\pointsto{c\dcomp b}w$ and $u\pointsto{a\dcomp b}w$ the definition of connector edge ensures 
that there is a connector edge $u\pointsto{a\dcomp b}q'$, for any such $q'$.  
Provided that $v\neq w$ an almost identical argument shows that there is a connector edge $v\pointsto{b}q'$.
\end{proof}

We need to verify the inductive hypothesis holds.  
It is the verification of demonic composition that requires particular attention, 
so we check the other details first.

First observe that domain is correctly represented, as this was already true in~$X$ 
and in the copy of~$F_{\D(s)}$, and each connector edge $\pointsto{a}$ started from a point 
in~$F_{\D(s)}$ that already had an outgoing edge $\pointsto{a}$.
Faithfulness is preserved trivially, as the representation on~$X$ was already faithful, 
and no new edges were added to this.  

Now we must check demonic composition.  

{\bf Compositional witness: if $x\pointsto{a\dcomp b}y$, find $z$ with $x\pointsto{a}z\pointsto{b}y$.}
Assume that $a\dcomp b$ labels some edge $x\pointsto{a\dcomp b}y$.  
We need to verify there is $z$ with $x\pointsto{a}z\pointsto{b}y$.  
If $x,y\in X$ or $x,y\in F_{\D(s)}$ then we are done, as $\dcomp $ is correctly represented on these sets.  
As there are no edges from $X$ to $F_{\D(s)}$, it remains to consider the case of a connector edge, 
where $x\in F_{\D(s)}$ and $y\in X$.  In this case, $a\dcomp b$ also labels an edge 
from $x\in F_{\D(s)}$ to some point $s\dcomp c\in F_s$ (with $x\neq s\dcomp c$), 
and $p\pointsto{c} y$ in $X$.  Because~$\dcomp $ is correctly represented in $F_{\D(s)}$ 
it follows that there is a point $z'\in F_{\D(s)}$ with $x\pointsto{a}z'\pointsto{b}s\dcomp c$.  
(In fact, $z'=x\dcomp a$ by the definition of the Wagner--Preston representation.)  
If $z'=s\dcomp c$ (implying $z'\in F_s$), then $x\dcomp a\dcomp b=x\dcomp a$, 
which by~\eqref{eqn:rightdegenerate} shows that $b=\D(b)$.  
Then $s\pointsto{c\dcomp b}s\dcomp c$ in $F_{\D(s)}$.  
So $c\dcomp b$ labels an edge starting at $p$ by Lemma~\ref{lem:copy}.  
Thus every edge labelled $c$ leaving $p$ in $X$ is followed by one labelled $b$; 
in particular this is true for the edge $p\pointsto{c}y$.  
As $b$ is a domain element, it follows that $b$ labels a loop at $y$.  
Thus $x\pointsto{a}y\pointsto{b}y$ so that the required $z$ can be chosen to be $y$.

Now assume that $z'\neq s\dcomp c$.  
In this instance, there is a connector edge $z'\pointsto{b}y$, so that we may choose $z$ to be $z'$.  
This completes the check for compositional witnesses.

{\bf Demonic witness: if $x\pointsto{a\dcomp b}y$, verify every $x\pointsto{a}z$ has $z$ in the domain of $b^\theta$.}
Assume $x\pointsto{a\dcomp b}y$ and that $x\pointsto{a}z$.  
Note that if $x\in X$ then so also are all of $x,y,z$ and we are done by the inductive hypothesis.  
So for the remainder of the proof it suffices to assume that $x\in F_{\D(s)}$.  
Now, if $y\in F_{\D(s)}$, then every edge in $F_{\D(s)}$ labelled $a$ leaving $x$ 
(and there is only one within $F_{\D(s)}$) can be followed by $b$.  
By Lemma~\ref{lem:copy2}, this is also true of every connector edge leaving $x$.   
As the edge $x\pointsto{a}z$ is either in $F_{\D(s)}$ or a connector edge, 
the verification is complete for when $y\in F_{\D(s)}$.  
Now assume that $y\in X$.  Then $x\pointsto{a\dcomp b}y$ is a connector edge, 
so there is a point $y'$ in $F_{s}$ with $x\pointsto{a\dcomp b}y'$.  
Then we are in the previous case and deduce that every edge labelled by $a$ leaving $x$ 
(be it in $F_{\D(s)}$ or a connector edge) is followed by one labelled $b$.

{\bf Composition: if $(x,z)\in a^\theta\dcomp b^\theta$, verify that $x\pointsto{a\dcomp b}z$.}  
Assume $x\pointsto{a}y\pointsto{b}z$ and 
every $z'$ with $x\pointsto{a}y'$ has $z'$ in the domain of $b^\theta$.  
We need to show that $a\pointsto{a\dcomp b}z$.  
If $x,y,z\in F_{\D(s)}$ or $x,y,z\in X$ then we are done, 
because demonic composition is correctly represented in $F_{\D(s)}$ and in $X$, 
and we did not change the domains of any elements when adding new edges.  
(Note that this is the case even if $x,y,z\in F_{\D(s)}$ but we consider some $y'\in X$ 
that happens to lie in the domain of $b^\theta$: 
it remains true that every $x\pointsto{a}y''$ in $F_{\D(s)}$ also is in the domain of $b^\theta$, 
so we would still have $x\pointsto{a\dcomp b}z$.  
Alternatively, use the fact that composition is functional in $F_{\D(s)}$.)  
Thus we may assume that $x\in F_{\D(s)}$ but at least one of $y,z\in X$.  
If $y\in F_{\D(s)}$ then $y\pointsto{b}z$ is a connector edge, and the definition of such edges implies 
that there is $z'\in F_{s}$ with $x\pointsto{a}y\pointsto{b}z'$.  
Then $x\pointsto{a\dcomp b}z'$, as composition in $F_{\D(s)}$ is functional.  
But then $x\pointsto{a\dcomp b}z$ also, as required.  
Thus we may assume that both $y$ and $z$ lie in $X$ 
(there are no edges from $X$ to $F_{\D(s)}$, so if $y$ in $X$ then $z\in X$ also).

In this instance, there is $y'\in F_s$ and $x\pointsto{a}y'$.  
We are assuming that every such point $y'$ is in the domain of $b^\theta$, 
so it follows that there is $z'$ with $y'\pointsto{b} z'$.  
Moreover, as $y'\in F_s$ we can select  $z'\in F_s$ also.  T
hen $x\pointsto{a\dcomp b}z'$, as composition in $F_{\D(s)}$ is functional.  
We have not yet shown that $x\pointsto{a\dcomp b}z$.  
There is $c\in S$ such that $s\pointsto{c}s\dcomp c=y'$ in $F_{\D(s)}$, 
and therefore $s\pointsto{c\dcomp b}z'$.  
Thus $(c\dcomp b)^\theta$ is defined at $p$ also, and hence we have $p\pointsto{c\dcomp b}z$ 
(as $p\pointsto{c}y\pointsto{b}z$).  
Thus the definition of connector edges ensures that $x\pointsto{a\dcomp b}z$ also.  
This completes the proof of Lemma~\ref{lem-repcf}.
\end{proof}
By Lemma \ref{lem-freecf}, Theorem~\ref{thm-maind} now follows immediately from Lemma \ref{lem-repcf}.

\section{Conclusion and open problems}
We have identified simple axiom systems that precisely capture the equational properties of angelically modelled programs with composition, domain, range and union and demonically modelled programs with (demonic) composition, domain and range.  This is part of a wider effort to place various logical frameworks for the formal reasoning about programs into simple algebraic settings.  In the case of demonic operations, in particular, this is in relative infancy and we hope that the present contribution will stimulate further work in this direction.

We have used the fact that the free algebras are free from cycles in both the angelic and
demonic cases, and we noted that representable algebras are not cycle-free in general.
Thus axiomatising the representation classes may require different methods.
Note that the classes $\repr(\acomp,\dom,\ran)$ and $\repr(\acomp,\dom,\ran,+)$ of algebras 
of binary relations 
in fact have no finite axiomatisation \cite{hirmik,jacsto:amon}.  
On the other hand, the class of partial maps under composition, domain and range has a finite 
axiomatisation~\cite{sch}. 
We conclude with an open problem.

\begin{question}
Does the class $\repr(\dcomp,\dom,\ran)$ 
have a finite axiomatisation?  
\end{question}


\bibliographystyle{amsplain}

\end{document}